\documentclass[journal,draftcls,onecolumn,12pt,twoside]{IEEEtran}

\usepackage{amscd, amssymb, amsthm, amsmath, color}
\usepackage{graphicx, psfrag, enumerate, multirow, algorithmic}
\usepackage[linesnumbered,ruled,vlined]{algorithm2e}

\newtheorem{theorem}{\bf Theorem}
\newtheorem{proposition}[theorem]{\bf Proposition}
\newtheorem{corollary}[theorem]{\bf Corollary}
\newtheorem{lemma}[theorem]{\bf Lemma}
\newtheorem{example}{Example}

\newcommand{\GP}{{\mathtt{GP}}}

\renewcommand{\d}{\displaystyle}

\long\def\symbolfootnote[#1]#2{\begingroup
\def\thefootnote{\fnsymbol{footnote}}\footnote[#1]{#2}\endgroup}

\begin{document}
\title{The Global Packing Number for an Optical Network}

\author{Yuan-Hsun Lo, Yijin Zhang,~\IEEEmembership{Member,~IEEE}, Wing Shing Wong,~\IEEEmembership{Fellow,~IEEE} and Hung-Lin Fu,~\IEEEmembership{Member,~IEEE}
\thanks{Partially supported by the National Natural Science Foundation of China (No. 61301107, 61174060), the Hong Kong RGC Earmarked Grant CUHK414012 and the Shenzhen Knowledge Innovation Program JCYJ20130401172046453 (Y.-H.~Lo, Y.~Zhang and W.~S.~Wong), and by Ministry of Science and Technology, Taiwan under grants MOST 104-2115-M-009-009 (H.-L.~Fu).}
\thanks{Y.-H. Lo is with the School of Mathematical Sciences, Xiamen University, Xiamen 361005, China. E-mail: yhlo0830@gmail.com}
\thanks{Y. Zhang is with the School of Electronic and Optical Engineering, Nanjing University of Science and Technology, Nanjing, China.
E-mail: yijin.zhang@gmail.com}
\thanks{W. S. Wong is with the Department of Information Engineering, The Chinese University of Hong Kong, Shatin, N.~T., Hong Kong.
E-mail: wswong@ie.cuhk.edu.hk}
\thanks{H.-L. Fu is with the Department of Applied Mathematics, National Chiao Tung University, Hsinchu 300, Taiwan, ROC. E-mail: hlfu@math.nctu.edu.tw}
}
\maketitle

\begin{abstract}
The {\em global packing number problem} arises from the investigation of optimal wavelength allocation in an optical network that employs Wavelength Division Multiplexing (WDM).
Consider an optical network that is represented by a connected, simple graph $G$. 
We assume all communication channels are bidirectional, so that all links and paths
are undirected.
It follows that there are ${|G|\choose 2}$ distinct node pairs associated with $G$,
where $|G|$ is the number of nodes in $G$.
A path system $\mathcal{P}$ of $G$ consists of ${|G|\choose 2}$ paths, one path to
connect each of the node pairs.  The global packing number of a path system $\mathcal{P}$, denoted by $\Phi(G,\mathcal{P})$, is the minimum integer $k$ to guarantee the existence of a mapping $\omega:\mathcal{P}\to\{1,2,\ldots,k\}$, such that $\omega(P)\neq\omega(P')$ if $P$ and $P'$ have common edge(s).
The global packing number of $G$, denoted by $\Phi(G)$, is defined to be the minimum $\Phi(G,\mathcal{P})$ among all possible path systems $\mathcal{P}$.
If there is no wavelength conversion along any optical transmission path for any node pair in the network, the global packing number signifies the minimum number of wavelengths required to support simultaneous communication for all pairs in the network.

In this paper, the focus is on ring networks, so that $G$ is a cycle.
Explicit formulas for the global packing number of a cycle is derived.
The investigation is further extended to chain networks.
A path system, $\mathcal{P}$, that enjoys $\Phi(G,\mathcal{P})=\Phi(G)$ is called {\em ideal}.  A characterization of ideal path systems is also presented.
We also describe an efficient heuristic algorithm to assign wavelengths that can be applied
to a general network with more complicated traffic load.

\end{abstract}


\section{Introduction} \label{sec:intro}


This paper investigates a class of resource allocation problem that deals with the computation of the {\em global packing number} of a communication network.  This
is an index that characterizes the number of wavelengths required to support a uniformly loaded optical network that employs Wavelength Division Multiplexing (WDM).  The global packing number problem is therefore well motivated by engineering application.

Let a graph, $G$, represent a communication network so that the link between any two nodes represents a bidirectional communication channel.  
For simplicity, we assume a uniform traffic model for which the
traffic load between any node pair is identical.   
A classical problem of path routing deals with the question of finding the shortest path between any node pair.   
On the other hand, some networking problems deal with the issue of distributing the routed traffic evenly over the network.   
If the links have identical traffic capacity, $\mathcal{C}$, one formulation of this latter class of questions is to find routing paths that they can all be supported simultaneously by the minimum $\mathcal{C}$. 
 
In an optic network, links between nodes are implemented by optical fibers that carry optic carrier signals.   
The WDM technology \cite{GY_04,LXC_05,S_10,SSA_02} allows carrier signals of different wavelengths be mixed together for simultaneous transmission over the optic fibers.  
In the basic implementation architecture, the wavelength of a carrier signal is kept constant as it passes over the network nodes.  
Hence, one can envision the end-to-end transmission from a source to a destination node as implemented over a reserved {\em lightpath}.  
Different simultaneous transmissions that go through a given optical fiber link must use different wavelengths.  
Hence, there is a basic question of computing the number of wavelengths required in a WDM network in order to support a given traffic model.  
This defines the global packing number problem for $G$ if we assume the traffic load
is simply one request for each distinct node pair. 
 
To facilitate subsequent discussion, we state the following notation and basic results.

\subsection{Notation} \label{sec:intro_notation}
Let $G$ be a connected simple graph. 
For any two nodes, $a,b\in V(G)$, one can assign a particular path to connect $a$ and $b$
and denote it by $P_{\{a,b\}}$.
Note that $P_{\{a,b\}}$ is sometimes set to be one of the shortest paths that connects $a$ with $b$.
Let $$\mathcal{P}:=\{P_{\{a,b\}}:\,a,b\in V(G)\}$$ denote a set of assigned paths,
one for each distinct node pair.
We say $\mathcal{P}$ is a \emph{path system} of $G$.
Obviously, $|\mathcal{P}|={|G|\choose 2}$, where $|G|$ is the number of nodes in $G$.
Let $\mathfrak{P}_G$ denote the collection of all path systems of $G$.
A path system is called a \emph{shortest path system} if every path is a shortest path connecting the two corresponding nodes.

Given a path system $\mathcal{P}\in\mathfrak{P}_G$ of $G$. 
A \emph{global packing} of $(G,\mathcal{P})$ is a mapping $\omega$ from $\mathcal{P}$ to the set $\{1,2,\ldots,k\}$, for some $k$, such that for any two paths $P,P'\in\mathcal{P}$, $\omega(P)\neq \omega(P')$ provided that $P$ and $P'$ have one or more than one edge in common.
Each element in the image $\{1,2,\ldots,k\}$ is called a \emph{wavelength}, named for its appliction to optical networks.
The \emph{global packing number} of $(G,\mathcal{P})$, denoted by $\Phi(G,\mathcal{P})$, is the minimum integer $k$ to guarantee the existence of a global packing.
The global packing number of $G$, denoted by $\Phi(G)$, is defined to be the smallest global packing number of $(G,\mathcal{P})$ among all path systems $\mathcal{P}\in\mathfrak{P}_G$; i.e., $$\Phi(G):=\min_{\mathcal{P}\in\mathfrak{P}_G}\Phi(G,\mathcal{P}).$$
A path system $\mathcal{P}$ with $\Phi(G,\mathcal{P})=\Phi(G)$ is said to be \emph{ideal}.
An ideal shortest path system is said to be {\em perfect}.

Consider a path system $\mathcal{P}\in\mathfrak{P}_G$.
In \emph{graph coloring} model, let $\mathcal{H}$ be the graph such that $V(\mathcal{H})=\mathcal{P}$, and $P$ is adjacent to $P'$ in $\mathcal{H}$ if and only if $P$ and $P'$ have the some edge(s) in common in $G$.
Then $\Phi(G,\mathcal{P})=\chi(\mathcal{H})$, the \emph{chromatic number} of $\mathcal{H}$.

All terminology and notation on graph theory used throughout this paper can be referred
to the textbook~\cite{West_01}.



\subsection{Bounds derived from basic principles} \label{sec:intro_bound}
Let $||G||$ denote the number of edges in the graph $G$.
A \emph{$G$-packing} is a collection of mutually edge-disjoint subgraphs of $G$.
A $G$-packing is \emph{full} if it contains $||G||$ edges.
Let $\mathcal{P}=\{P_1,P_2,\ldots,P_m\}$ be a path system of $G$, where $m={|G|\choose 2}$, and $\omega$ be a global packing of $(G,\mathcal{P})$.
Then, for each $i$, the preimage of $i$ under $\omega$, $\omega^{-1}(i)$, forms a $G$-packing.
Since a $G$-packing consists of at most $||G||$ edges, a natural lower bound of $\Phi(G,\mathcal{P})$ is given as follows.

\begin{proposition}\label{pro:lowerbound}
Let $\mathcal{P}\in\mathfrak{P}_G$ be a path system of $G$.
Then,
$$\Phi(G,\mathcal{P}) \geq \left\lceil \frac{\sum_{P\in\mathcal{P}} ||P||}{||G||} \right\rceil.$$
\end{proposition}

On the other hand, assume that $\mathcal{P}$ can be partitioned into $k$ subsets, $\mathcal{S}_1,\ldots,\mathcal{S}_k$, each of which forms a $G$-packing. 
One can define a global packing by $\omega:=P\mapsto i$ if $P\in\mathcal{S}_i$ for all $i$. 
It follows that $\Phi(G,\mathcal{P})\leq k$.

\begin{proposition}\label{pro:upperbound}
Let $\mathcal{P}\in\mathfrak{P}_G$ be a path system of $G$.
If $\mathcal{P}$ can be partitioned into $k$ subsets such that each subset forms a $G$-packing, then
$$\Phi(G) \leq \Phi(G,\mathcal{P}) \leq k.$$
\end{proposition}

In this paper, the focus is on ring networks, so that $G$ is a cycle.
In spite of its simplicity, the ring topology is of fundamental interest in the study of optical networks~\cite{GCFPFGNP_01,GSCAC_09,PCF_05}.
Let $C_n$ denote a cycle of $n$ nodes.
For the sake of convenience, in this paper we assume the nodes of $C_n$ are labelled clockwise by non-negative integers $0,1,\ldots,n-1$, and $\langle a,b\rangle$ denotes the path consists of edges $\{a,a+1\},\{a+1,a+2\},\ldots,\{b-1,b\}$ (mod $n$).
Note that $\langle a,b\rangle=\{a,b\}$ if $b=a+1$.

The rest of this paper is organized as follows.
In Section~\ref{sec:evencycle} we compute the exact value of $\Phi(C_n)$ for the case $n$ is even by means of Proposition~\ref{pro:lowerbound} and Proposition~\ref{pro:upperbound}. 
As for odd $n$ cases, the exact value of $\Phi(C_n)$ is given in Section~\ref{sec:oddcycle}, where the derivation is based on a greedy construction, which we refer to as
{\em Intelligent Packing algorithm}.
The corresponding ideal path systems of $C_n$ are characterized as well.
In Section~\ref{sec:simulation} we extend the Intelligent Packing algorithm to a general version, referred to as the {\em Length First Packing algorithm}, which is applicable to
general connected graphs.  To demonstrate the effectiveness of the Length First
Packing algorithm we apply it to chain networks and show that they can achieve the global packing number value. 
We also apply it to ring networks with more complicated traffic
load and compare its performance against a random algorithm.
Finally, concluding remarks are given in Section~\ref{sec:conclusion}.

\section{Global packing numbers of even cycles}\label{sec:evencycle}
To derive a lower bound of $\Phi(C_{2n})$, by Proposition~\ref{pro:lowerbound}, it is natural to consider the shortest path system.
Fix two nodes $a,b$ in a cycle $C_{2n}$.
If the distance of $a$ and $b$ is less than $n$, there is only one shortest path that connects them; however, if the distance is equal to $n$ (i.e., $a$ and $b$ are on opposite side), there are two choices of shortest paths.

\begin{lemma}\label{lemma:cycle_even_lower}
For any integer $n\geq 1$, $$\Phi(C_{2n})\geq {n\choose 2} + \left\lfloor\frac{n}{2}\right\rfloor +1.$$
\end{lemma}
\begin{proof}
Let $\mathcal{P}$ be a shortest path system of $C_{2n}$.
Since $$\sum_{P\in\mathcal{P}}||P|| \leq \sum_{P\in\mathcal{P'}}||P||$$ for any path system $\mathcal{P'}\in\mathfrak{P}_{C_{2n}}$ with $\mathcal{P'}\neq\mathcal{P}$, we have $\Phi(C_{2n},\mathcal{P}')\geq\Phi(C_{2n},\mathcal{P})$.
Observe that there are exactly $n$ paths of length $n$ and $2n$ paths of length $i$, for each $i=1,2,\ldots,n-1$. 
By Proposition~\ref{pro:lowerbound}, 
\begin{equation}\label{eq:lower_even_1}
\begin{split}
\Phi(C_{2n}) &= \min_{\mathcal{P}'\in\mathfrak{P}_G}\Phi(G,\mathcal{P}') \geq\Phi(C_{2n},\mathcal{P}) \\
&\geq \left\lceil\frac{n\cdot n + \sum_{i=1}^{n-1}2ni}{2n} \right\rceil = {n\choose 2} + \left\lceil\frac{n}{2}\right\rceil. 
\end{split}
\end{equation}
Hence it follows that $\Phi(C_{2n})\geq {n\choose 2} + \lfloor\frac{n}{2}\rfloor +1$ when $n$ is odd.
In what follows, we consider the case $n$ is even.

Let $\mathcal{P}^*$ be one of the ideal path systems of $C_{2n}$.
There are two cases: $\mathcal{P}^*$ is or is not a shortest path system.
If $\mathcal{P}^*$ is not a shortest path system, by the same argument in \eqref{eq:lower_even_1}, $\Phi(C_{2n},\mathcal{P}^*)$ must be larger than ${n\choose 2}+\frac{n}{2}$, and then we are done.
If $\mathcal{P}^*$ is a shortest path system, we now claim $\Phi(C_{2n},\mathcal{P}^*)\geq{n\choose 2}+\frac{n}{2}+1$.
Suppose not, i.e., $\Phi(C_{2n},\mathcal{P}^*)\leq {n\choose 2}+\frac{n}{2}$.
By \eqref{eq:lower_even_1}, it must be $\Phi(C_{2n},\mathcal{P}^*)={n\choose 2}+\frac{n}{2}$.
Fix an edge $e$, and consider the number of paths of length less than $n$ which contain $e$.
Then there are ${n\choose 2}$ paths: one of length $1$, two of length $2$, $\ldots$, and $n-1$ of length $n-1$.
Since all of them contain $e$, they have to be assigned distinct wavelengths.
That is, $\Phi(C_{2n},\mathcal{P}^*)\geq {n\choose 2}+x$, where $x$ is the number of paths of length $n$ which contain $e$.
Note that this fact is true for all distinct edges.
Let $\mathcal{C}$ be the collection of $n$ paths of length $n$ in $\mathcal{P}^*$.
Now, it suffices to prove that there exists an edge which occurs in $\mathcal{C}$, at least $\frac{n}{2}+1$ times.

\begin{figure}[h]
\centering
\includegraphics[width=2.5in]{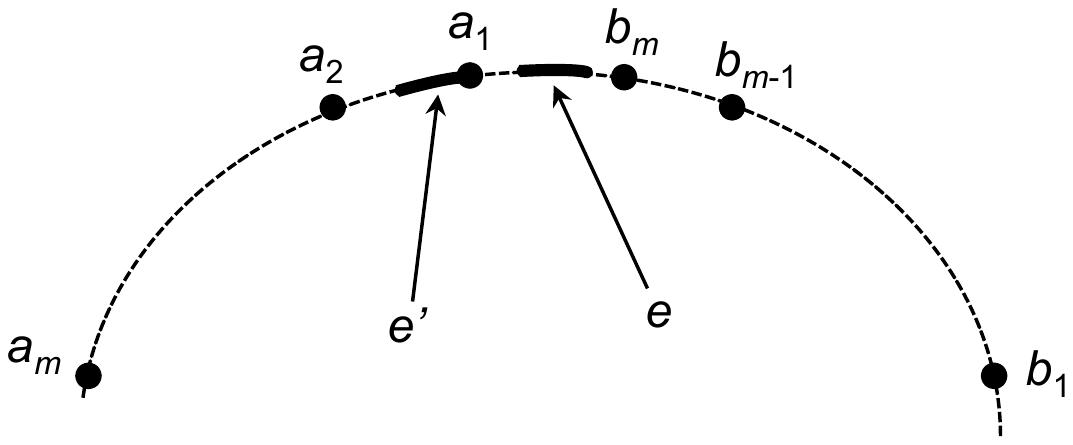}
\caption{The $m$ paths of length $n$ which contain $e$.} 
\label{fig:even-proof}
\end{figure}

Suppose the assertion is not true.
Since there are $2n$ edges in $C_{2n}$, each edge occurs in $\mathcal{C}$ exactly $\frac{n}{2}$ times.
Denote $m=\frac{n}{2}$ and let $e$ be an edge which occurs in paths $\langle a_1,b_1\rangle$, $\langle a_2,b_2\rangle$, $\ldots, \langle a_m,b_m\rangle$.
For convenience, let $a_m\prec\cdots\prec a_2\prec a_1$ and $b_m\prec\cdots\prec b_2\prec b_1$, in which $a\prec b$ denotes that $a$ is ahead of $b$ counterclocksiwely, as illustrated in Fig.~\ref{fig:even-proof}.
Hence $e\subseteq\langle a_1,b_m\rangle$, and all the edges in $\langle a_1,b_b\rangle$ occur $m$ times as well.
Consider the edge $e'=\{a_1-1,a_1\}$ (note that $a_1-1$ may be $a_2$).
Clearly, in these $m$ paths, $e'$ only occurs $m-1$ times.
Since $a_1$ is one of the endpoints of $\langle a_1,b_1\rangle$, it can not be an endpoint of the other paths of length $n$.
This implies that either $e'$ does not occur in the other paths of length $n$ or if it occurs, some of the edges in $\langle a_1,b_b\rangle$ will occur more than $m$ times, a contradiction.
For the first case, $e'$ occurs exactly $m-1$ times among all paths of length $n$, a contradiction.
Hence we complete the proof.
\end{proof}

Now, we will show that $\Phi(C_{2n}) = {n\choose 2}+\lfloor\frac{n}{2}\rfloor+1$ by choosing some one particular shortest path system $\mathcal{P}$.
As an illustration, consider the following example.

\begin{example}\label{exa:20}
{\rm
Let $n=9$ and $\mathcal{P}$ be the shortest path system of $C_{18}$ which contains the following $9$ paths of length $9$:
$$
\begin{array}{lllll}
\langle 0,9\rangle, & \langle 2,11\rangle, & \langle 4,13\rangle, & \langle 6,15\rangle, & \langle 8,17\rangle, \\
\langle 10,1\rangle, & \langle 12,3\rangle, & \langle 14,5\rangle, & \langle 16,7\rangle. & 
\end{array}
$$
Except for above $9$ paths, in $\mathcal{P}$ there are $18$ paths of length $i$, for each $i=1,2,\ldots,8$.
By Proposition~\ref{pro:upperbound}, we aim to partition these $153$ paths into ${9\choose 2}+\lfloor\frac{9}{2}\rfloor+1=41$ $C_{18}$-packings.
First, partition the set $\{1,2,\ldots,8\}$ into subsets $\{1,8\}$, $\{2,7\}$, $\{3,6\}$ and $\{4,5\}$, so that the sum of all elements of each subset is $9$, a factor of $18$.
For each subset, we construct $9$ full $C_{18}$-packings, which consist of all paths with lengths belonging to the subset.
For example, the $9$ full $C_{18}$-packings produced from $\{1,8\}$ are listed below.
$$
\begin{array}{rcccc}
B_0: & \langle 0,1\rangle & \langle 1,9\rangle & \langle 9,10\rangle & \langle 10,0\rangle \\
B_1: & \langle 1,2\rangle & \langle 2,10\rangle & \langle 10,11\rangle & \langle 11,1\rangle \\
B_2: & \langle 2,3\rangle & \langle 3,11\rangle & \langle 11,12\rangle & \langle 12,2\rangle \\
B_3: & \langle 3,4\rangle & \langle 4,12\rangle & \langle 12,13\rangle & \langle 13,3\rangle \\
B_4: & \langle 4,5\rangle & \langle 5,13\rangle & \langle 13,14\rangle & \langle 14,4\rangle \\
B_5: & \langle 5,6\rangle & \langle 6,14\rangle & \langle 14,15\rangle & \langle 15,5\rangle \\
B_6: & \langle 6,7\rangle & \langle 7,15\rangle & \langle 15,16\rangle & \langle 16,6\rangle \\
B_7: & \langle 7,8\rangle & \langle 8,16\rangle & \langle 16,17\rangle & \langle 17,7\rangle \\
B_8: & \langle 8,9\rangle & \langle 9,17\rangle & \langle 17,0\rangle & \langle 0,8\rangle
\end{array}
$$

\noindent
These $9$ full $C_{18}$-packings can be viewed as the results of cyclic rotation of the first one, $B_0$, called the \emph{base packing}.
More precisely, we first cover the edges of $C_{18}$ with two paths of length $1$ and two paths of length $8$ alternatively to form the base packing, and then rotate it to obtain the others.
This ensures that there are no repeated paths among these rotations.
Note that this construction works for the other three subsets: $\{2,7\}$, $\{3,6\}$ and $\{4,5\}$.
We have a total of $36$ full $C_{18}$-packings, which consist of all paths of length $1,2,\ldots,8$.

Next, we shall replace some paths in the full $C_{18}$-packings produced from the subset $\{1,8\}$ with the paths of length $9$.
In $B_0$, $\langle 0,1\rangle$ and $\langle 1,9\rangle$ are replaced by $\langle 0,9\rangle$; in $B_1$, $\langle 10,11\rangle$ and $\langle 11,1\rangle$ are replaced by $\langle 10,1\rangle$; in $B_2$, $\langle 2,3\rangle$ and $\langle 3,11\rangle$ are replaced by $\langle 2,11\rangle$; in $B_3$, $\langle 12,13\rangle$ and $\langle 13,3\rangle$ are replaced by $\langle 12,3\rangle$; and so on.
Following this procedure, in $B_8$, $\langle 8,9\rangle$ and $\langle 9,17\rangle$ are replaced by $\langle 8,17\rangle$.
Finally, those $18$ paths ($9$ of which are of length $1$ and the others are of length $8$) that are taken off will form $5$ $C_{18}$-packings as below.

$$
\begin{array}{rcccccc}
D_0: & \langle 1,9\rangle & \langle 10,11\rangle & \langle 11,1\rangle \\
D_1: & \langle 3,11\rangle & \langle 12,13\rangle & \langle 13,3\rangle \\
D_2: & \langle 5,13\rangle & \langle 14,15\rangle & \langle 15,5\rangle \\
D_3: & \langle 7,15\rangle & \langle 16,17\rangle & \langle 17,7\rangle \\
D_4: & \langle 0,1\rangle & \langle 2,3\rangle & \langle 4,5\rangle & \langle 6,7 \rangle & \langle 8,9 \rangle & \langle 9,17 \rangle
\end{array}
$$
\noindent
Hence $\mathcal{P}$ is completely partitioned into $36+5=41$ $C_{18}$-packings, which achieves the bound of Lemma~\ref{lemma:cycle_even_lower} for $n=9$.
For a general proof, we need the following.
}
\end{example}

\begin{lemma}\label{lemma:partition}
Let $n$ be a positive integer, and let $X=\{x_1,\ldots,x_k\}$ be a set of positive integers with $x_i<\frac{n}{2}$ for all $i$.
If $n=tm$ for some positive integer $t$, where $m=\sum_{i=1}^k x_i$, then the collection of all paths of lengths belonging to $X$ on the cycle $C_n$ can be partitioned into $m$ full $C_n$-packings.
\end{lemma}
\begin{proof}
First, we construct a base $C_n$-packing by periodically putting a path of length $x_1$, a path of length $x_2$, $\ldots$, a path of length $x_k$ one after another on $C_n$ in clockwise direction.
Since $n=tm$, the base $C_n$-packing consists of $t$ paths of length $x_i$, for each $i$, and then clearly is a full packing.
Then, rotate this base $C_n$-packing step by step clockwise to produce the other $m-1$ full $C_n$-packings.
Since any two paths of length $x_i$, for each $i$, are completely overlapped if and only if the base $C_n$-packing is rotated by $m$ or a multiple of $m$ steps, these $m$ full $C_n$-packings have no repeated paths and thus form a partition of the set of all paths of lengths belonging to $X$.
\end{proof}

We are ready for the main result in this section.

\begin{theorem}\label{thm:cycle_even}
For any integer $n\geq 1$, $$\Phi(C_{2n})={n\choose 2} + \left\lfloor\frac{n}{2}\right\rfloor + 1.$$
\end{theorem}
\begin{proof}
By Lemma~\ref{lemma:cycle_even_lower} and Proposition~\ref{pro:upperbound}, it suffices to prove that there exists an ideal path system, $\mathcal{P}$, of $C_{2n}$ such that all paths of $\mathcal{P}$ can be partitioned into ${n\choose 2} + \left\lfloor\frac{n}{2}\right\rfloor + 1$ $C_{2n}$-packings.
Let $\mathcal{P}$ be the shortest path system of $C_{2n}$ which contains paths $\langle 0,n\rangle, \langle 2,n+2\rangle, \ldots, \langle 2n-2,n-2\rangle$ if $n$ is odd, or $\langle 0,n\rangle, \langle 2,n+2\rangle, \ldots, \langle n-2,2n-2\rangle, \langle n+1,1\rangle, \langle n+3,3\rangle, \ldots, \langle 2n-1,n-1\rangle$ otherwise.

Consider the case that $n$ is odd.
By Lemma~\ref{lemma:partition}, for $1\leq i\leq\lfloor\frac{n}{2}\rfloor$, all paths with lengths belonging to $\{i,n-i\}$ can be partitioned into $n$ full $C_{2n}$-packings.
Since $n\lfloor\frac{n}{2}\rfloor={n\choose 2}$ in the case $n$ is odd, the paths of length less than $n$ will produce ${n\choose 2}$ full $C_{2n}$-packings.
Now, consider the $n$ full $C_{2n}$-packings associated with the path lengths $1$ and $n-1$.
For $i=0,1,\ldots,n-1$, denote by $B_i$ the full $C_{2n}$-packing that contains the paths $\langle i,i+1\rangle, \langle i+1,i+n\rangle, \langle i+n,i+n+1\rangle$ and $\langle i+n+1,i\rangle$.
For each $i$, if $i$ is even, replace the two paths $\langle i,i+1\rangle, \langle i+1,i+n\rangle$ of $B_i$ with $\langle i,i+n\rangle$; otherwise, replace the two paths $\langle i+n,i+n+1\rangle, \langle i+n+1,i\rangle$ of $B_i$ with $\langle i+n,i\rangle$.
See Fig.~\ref{fig:replacement} for an illustration.
Let $B_i'$ be the set of paths in $B_i$ which are replaced by a path of length $n$, for $0\leq i\leq n-1$.
More precisely, $B_i'=\{\langle i,i+1\rangle, \langle i+1,i+n\rangle\}$ if $i$ is even, and $B_i'=\{\langle i+n,i+n+1\rangle, \langle i+n+1,i\rangle\}$ otherwise.
The paths in $\mathcal{B}:=\bigcup_{i=0}^{n-1} B_i'$ are rearranged as follows.
For $i=0,1,\ldots,\lfloor\frac{n}{2}\rfloor-1$, let $$D_i=B_{2i}\cup B_{2i+1}\setminus \{\langle 2i,2i+1\rangle\}.$$ 
And, let $$D_{\lfloor\frac{n}{2}\rfloor}=\left\{\langle 2i,2i+1\rangle:\,i=0,1,\ldots,\left\lfloor\frac{n}{2}\right\rfloor-1\right\}\cup B_{n-1}'.$$
It is easy to see that $D_0,D_1,\ldots,D_{\lfloor\frac{n}{2}\rfloor}$ form a partition of $\mathcal{B}$, and each of them is a $C_{2n}$-packing.
Hence we ultimately partition $\mathcal{P}$ into ${n\choose 2}+\lfloor\frac{n}{2}\rfloor+1$ $C_{2n}$-packings.

\begin{figure*}
\centering
\includegraphics[width=6in]{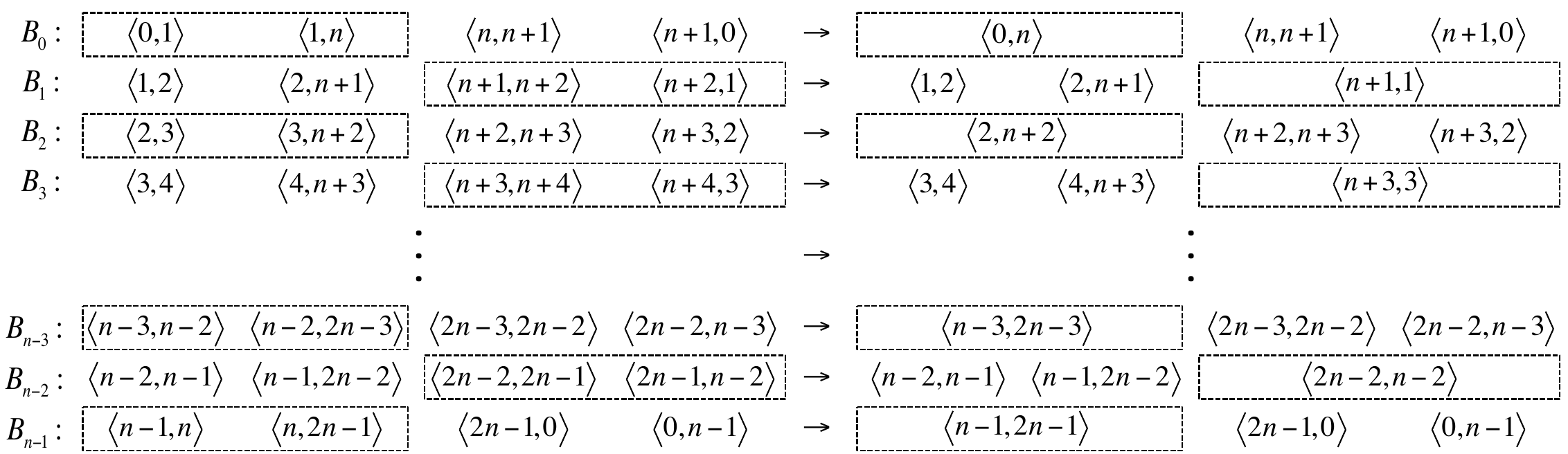}
\caption{Replace some paths of $B_i$ with a path of length $n$.} \label{fig:replacement}
\end{figure*}

The case $n$ is even can be dealt with in a similar fashion.
We omit the detail, and the proof is completed.
\end{proof}

We note here that `shortest' is not a necessary condition for a path system to be ideal.
In $C_4$, the following path system, see Table~I, is ideal but not shortest, for instance.
The path marked by $*$ indicates a non-shortest path.

\begin{table}\label{tab:counterexample}
\begin{tabular}{|c|c|c|}
\hline
node pair &
\begin{tabular}{c}
path \\
system
\end{tabular} & 
\begin{tabular}{c}
receiving \\
wavelength
\end{tabular} \\ \hline \hline
$\{0,1\}$ & ${}^*\langle 1,0\rangle$ & 1 \\ \hline
$\{0,2\}$ & $\langle 0,2\rangle$ & 2 \\ \hline
$\{0,3\}$ & $\langle 3,0\rangle$ & 2 \\ \hline
$\{1,2\}$ & $\langle 1,2\rangle$ & 3 \\ \hline
$\{1,3\}$ & $\langle 3,1\rangle$ & 3 \\ \hline
$\{2,3\}$ & $\langle 2,3\rangle$ & 2 \\ \hline
\end{tabular}
\caption{An example that a non-shortest path system is ideal.}
\end{table}
In the case that $n$ is odd, however, an ideal path systems of $C_{2n}$ must be perfect.

\begin{corollary}\label{coro:even}
Let $n>1$ be an odd integer.
If $\mathcal{P}$ is an ideal path system of $C_{2n}$, then $\mathcal{P}$ is perfect.
\end{corollary}
\begin{proof}
Suppose $\mathcal{P}$ is a non-shortest path system of $C_{2n}$ satisfying that $\Phi(\mathcal{P},C_{2m})=\Phi(\mathcal{P})$.
Assume that there are $t$ non-shortest paths in $\mathcal{P}$, and let $\mathcal{P}=\mathcal{P}^{(t)}\uplus\mathcal{P}_n$, where $\mathcal{P}_n$ denotes the collection of all paths of length $n$ and $\uplus$ refers to the disjoint union operation.
Let $\mathcal{P}^{(0)}$ denote the set of shortest paths of length less than $n$ in $C_{2n}$.
Then, $$\mathcal{P}^{(0)}\leadsto\mathcal{P}^{(1)}\leadsto\mathcal{P}^{(2)}\leadsto\cdots\leadsto\mathcal{P}^{(t)},$$ where $\leadsto$ is the operation that replaces some shortest path with its corresponding non-shortest path having the same endpoints.

Fix a pair of antipodal edges $e,e'$, and let $N_i(e)$ and $N_i(e')$ be the number of paths in $\mathcal{P}^{(i)}$ which contain $e$ and $e'$, respectively, for $i=0,1,\ldots,t$.
By \eqref{eq:lower_even_1}, it is obvious that $N_0(e)=N_0(e')={n\choose 2}$. 
Let $P\in\mathcal{P}^{(i)}\setminus\mathcal{P}^{(i+1)}$ and $\widehat{P}\in\mathcal{P}^{(i+1)}\setminus\mathcal{P}^{(i)}$ with the same endpoints, i.e., $P$ is a shortest path while $\widehat{P}$ is not.
Since $P$ contains at most one of $e$ and $e'$, there are two cases as follows.
If $P$ contains one of them, say $e$, then $N_{i+1}(e)=N_i(e)-1$ and $N_{i+1}(e')=N_i(e')+1$.
If $P$ contains neither $e$ nor $e'$, then $N_{i+1}(e)=N_i(e)+1$ and $N_{i+1}(e')=N_i(e')+1$.
This concludes that $N_t(e)+N_t(e')=2{n\choose 2}+2x$, for some integer $x\geq 0$.
Note that this fact is true for all distinct pairs of antipodal edges.
Since $\mathcal{P}^{(t)}$ contains some non-shortest paths, by \eqref{eq:lower_even_1} there exists some pair of antipodal edges $e,e'$ such that $N_t(e)+N_t(e')\geq 2{n\choose 2}+2$.
In addition, since either $e$ or $e'$ occurs on any paths in $\mathcal{P}_n$, one of $e$ and $e'$ is contained in at least ${n\choose 2}+\lceil\frac{n}{2}\rceil+1={n\choose 2}+\lfloor\frac{n}{2}\rfloor+2$ distinct paths in $\mathcal{P}^{(t)}\uplus\mathcal{P}_n(=\mathcal{P})$.
Hence the path system $\mathcal{P}$ needs at least ${n\choose 2}+\lfloor\frac{n}{2}\rfloor+2$ wavelengths, and thus is not an ideal path system due to Theorem~\ref{thm:cycle_even}.
This completes the proof.
\end{proof}

\section{Global packing numbers of odd cycles}\label{sec:oddcycle}
In this section we deal with the global packing number of $C_{2n+1}$. 
Similar to Section~\ref{sec:evencycle}, we use shortest path systems to estimate the lower bound of $\Phi(C_{2n+1})$.
Notice that in $C_{2n+1}$ the shortest path between any two nodes is unique, that is,
$$P_{\{a,b\}}=
\begin{cases}
\langle a,b\rangle, & \text{ if } a-b>n \text{ (mod $2n+1$),} \\
\langle b,a\rangle, & \text{ if } a-b\leq n \text{ (mod $2n+1$).}
\end{cases}
$$
Hence there is no confusion on choosing the shortest path system.

\begin{lemma}\label{lemma:cycle_odd_lower}
For any integer $n\geq 1$, $$\Phi(C_{2n+1})\geq {n+1\choose 2}.$$
\end{lemma}
\begin{proof}
Let $\mathcal{P}$ be the shortest path system of $C_{2n+1}$.
In $\mathcal{P}$ there are exactly $2n+1$ paths of length $i$, for $i=1,2,\ldots,n$. 
Since $\sum_{P\in\mathcal{P}'}||P|| \geq \sum_{P\in\mathcal{P}}||P||$ for any path system $\mathcal{P}'\in\mathfrak{P}_{C_{2n+1}}$, by Proposition~\ref{pro:lowerbound}, we have
\begin{align*}
\Phi(C_{2n+1}) &= \min_{\mathcal{P}'\in\mathfrak{P}_{C_{2n+1}}}\Phi(C_{2n+1},\mathcal{P}')\geq\Phi(C_{2n+1},\mathcal{P}) \\
&\geq \left\lceil\frac{\sum_{i=1}^n(2n+1)i}{2n+1} \right\rceil = {n+1\choose 2}.
\end{align*}
\end{proof}

Now, we introduce an algorithm, named \emph{Intelligent Packing (IP)}, to produce a global packing for the shortest path system of $C_{2n+1}$.
We use a $(2n+1)\times (2n+1)$ symmetric array, denoted by $\GP$, to represent the resulting global packing $\omega$, i.e., $\GP(a,b)$ and $\GP(b,a)$ indicate the value $\omega(P_{\{a,b\}})$.
The addition or subtraction herein is taken modulo $2n+1$.

\medskip

\begin{algorithm}
  \SetAlgoLined
  \SetKwInOut{Input}{input}\SetKwInOut{Output}{output}
  \Input{the shortest path system of $C_{2n+1}$}
  \Output{the maximal visited wavelength index, \emph{total}}
  initialization: $total=0$, $\ell=n$, and $\GP(a,b)=0$ for all pairs $(a,b)$ \;
  \While{$\ell>0$}
  {
    \For{$i=0,1,2,\ldots,2n$}
 	{
 		\If{$\GP(i,i+\ell)=0$}{$\GP(i,i+\ell),\GP(i+\ell,i) \gets$ the least available wavelength index\;}
 		\If{$\GP(i,i-\ell)=0$}{$\GP(i,i-\ell),\GP(i-\ell,i) \gets$ the least available wavelength index\;}
 	}
 	update \emph{total} \;
    $\ell\gets \ell-1$\;
  }
    \Return \emph{total} \;
  \caption{Intelligent Packing (IP)}
  \label{alg:GGP}
\end{algorithm}

\medskip

The main idea of IP is to greedily assign the smallest free wavelength index to paths of $\mathcal{P}$ in the descending order of the path length.
For the paths with the same length, we arrange them in the order of their smaller endpoint labels.
After all paths of lengths lager than or equal to $\ell$ receiving their wavelength indices, denote by $T_{\ell}$ the maximal visited wavelength index, and by $\mathcal{L}_{\ell}(k)$ the set of edges in which the wavelength indexed $k$ is free, for $1\leq k\leq T_{\ell}$.
For the sake of convenience, elements of $\mathcal{L}_{\ell}(k)$, called \emph{idle bands} on wavelength $k$ after the $\ell$-path round, are represented as ordered pairs $(s,t)$, indicating a series of adjacent edges $\{s,s+1\}, \{s+1,s+2\}, \ldots, \{t-1,t\}$ (mod $2n+1$), whenever they are all in $\mathcal{L}_{\ell}(k)$.
As an illustration, consider the following example: $n=5$.

\begin{example}
{\rm 
Let $\mathcal{P}$ be the shortest path system of $C_{11}$.
In the first round ($\ell=5$) of the IP algorithm, the $11$ paths of length $5$ are arranged in the order:  $\langle 0,5\rangle$, $\langle 6,0\rangle$, $\langle 1,6\rangle$, $\langle 7,1\rangle$, $\langle 2,7\rangle$, $\langle 8,2\rangle$, $\langle 3,8\rangle$, $\langle 9,3\rangle$, $\langle 4,9\rangle$, $\langle 10,4\rangle$ and $\langle 5,10\rangle$.
Then, paths $\langle 0,5\rangle$ and $\langle 6,0\rangle$ receive wavelength index $1$, paths $\langle 1,6\rangle$ and $\langle 7,1\rangle$ receive wavelength index $2$,  $\ldots$, path $\langle 5,10\rangle$ receives wavelength index $6$.
Fig.~\ref{fig:11}(a) is a graphic representation of the wavelength assignment.
In Fig.~\ref{fig:11}, the $x$- and $y$-axis respectively represent the node labels and wavelength indices, and an $1\times\ell$ rectangle with bold boundary on the $k$-layer indicates that the corresponding path of length $\ell$ receives wavelength index $k$.
In addition, the circled numbers represent the order of edges in receiving wavelength indices in each round.
One can see that the maximal visited wavelength index $T_5=6$, and the corresponding idle bands $\mathcal{L}_5(k)=\{(k+4,k+5)\}$, for $1\leq k\leq 5$, and $\mathcal{L}_5(6)=\{(10,5)\}$.

\begin{figure*}
\centering
\includegraphics[width=6in]{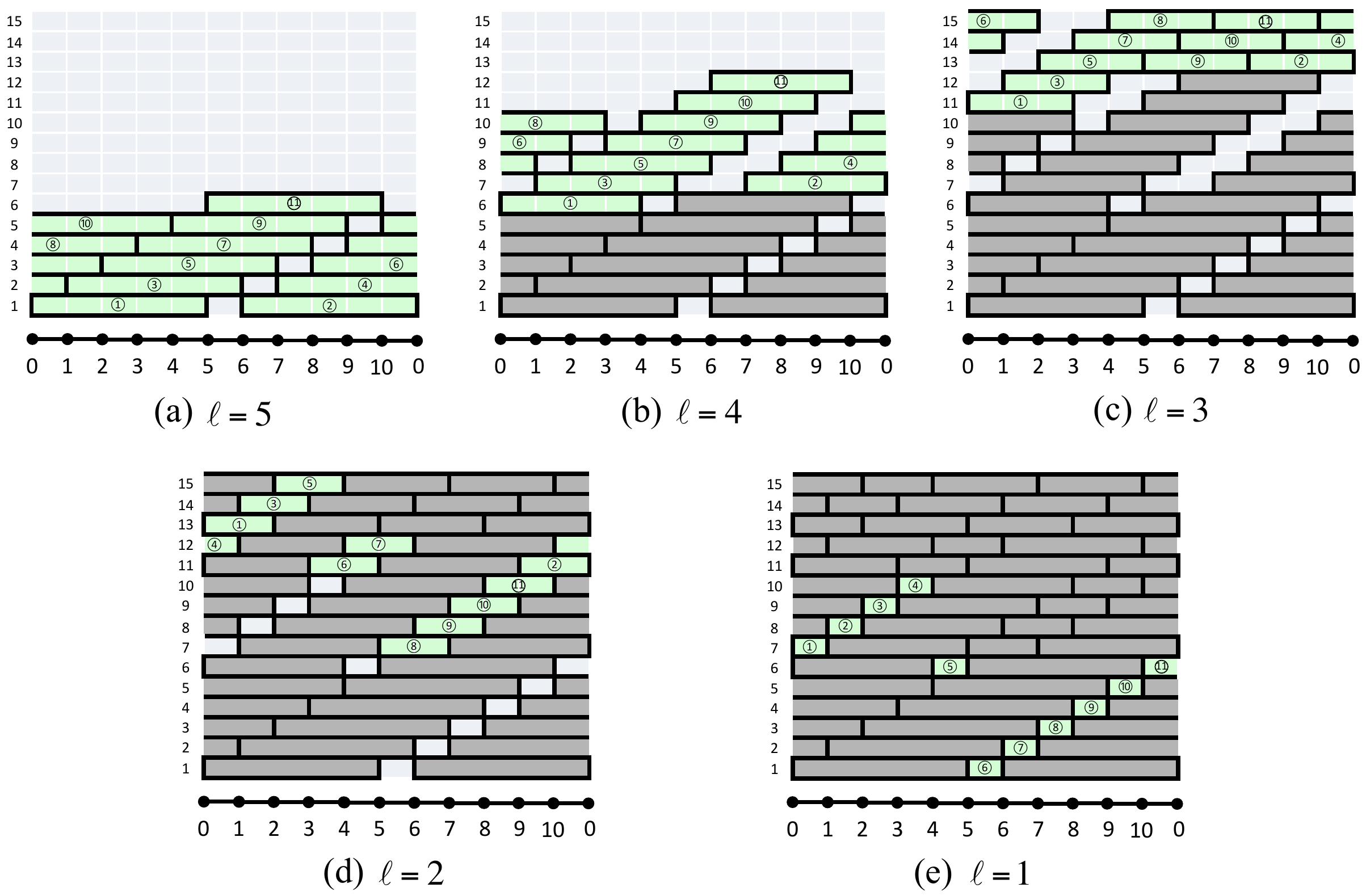}
\caption{A round-by-round illustration of Algorithm~\ref{alg:GGP} for $n=5$.} \label{fig:11}
\end{figure*}

In the second round ($\ell=4$), the $11$ paths of length $4$ are arranged in the order:  $\langle 0,4\rangle$, $\langle 7,0\rangle$, $\langle 1,5\rangle$, $\langle 8,1\rangle$, $\langle 2,6\rangle$, $\langle 9,2\rangle$, $\langle 3,7\rangle$, $\langle 10,3\rangle$, $\langle 4,8\rangle$, $\langle 5,9\rangle$ and $\langle 6,10\rangle$.
Under the wavelength assignment of the the paths of length $5$, path $\langle 0,4\rangle$ can receive wavelength index $6$, path $\langle 7,0\rangle$ can receive wavelength index $7$, path $\langle 1,5\rangle$ can receive wavelength index $7$, and the others' wavelength indices are illustrated in Figure~\ref{fig:11}(b).
After the first two rounds, one can see that $T_4=12$, and the corresponding idle bands are listed as follows.
\begin{itemize}
\item $\mathcal{L}_4(k)=\mathcal{L}_5(k)$, for $1\leq k\leq5$.
\item $\mathcal{L}_4(6)=\{(4,5),(10,0)\}$.
\item $\mathcal{L}_4(k)=\{(k+4,k+5),(k-2,k)\}$ for $7\leq k\leq10$.
\item $\mathcal{L}_4(k)=\{(k-2,k+5)\}$ for $11\leq k\leq 12$.
\end{itemize}
The remaining three rounds ($\ell=3,2$ and $1$) are shown in Figure~\ref{fig:11}(c) -- (e).
The maximal visited wavelength index is $15={6\choose 2}$, attaining the bound proposed in Lemma~\ref{lemma:cycle_odd_lower}.
This implies $\Phi(C_{11})=15$.
The resulting $\GP$ is listed in Table~II.  
}
\end{example}

\begin{table}[h]\label{tab:C11}
\scriptsize
\begin{tabular}{c|ccccccccccc}
 & 0 & 1 & 2 & 3 & 4 & 5 & 6 & 7 & 8 & 9 & 10 \\ \hline
0 & $-$ & 7 & 13 & 11 & 6 & 1 & 1 & 7 & 13 & 11 & 6 \\
1 & 7 & $-$ & 8 & 14 & 12 & 7 & 2 & 2 & 8 & 14 & 12 \\
2 & 13 & 8 & $-$ & 9 & 15 & 13 & 8 & 3 & 3 & 9 & 15 \\
3 & 11 & 14 & 9 & $-$ & 10 & 11 & 14 & 9 & 4 & 4 & 10 \\
4 & 6 & 12 & 15 & 10 & $-$ & 6 & 12 & 15 & 10 & 5 & 5 \\
5 & 1 & 7 & 13 & 11 & 6 & $-$ & 1 & 7 & 13 & 11 & 6 \\
6 & 1 & 2 & 8 & 14 & 12 & 1 & $-$ & 2 & 8 & 14 & 12 \\
7 & 7 & 2 & 3 & 9 & 15 & 7 & 2 & $-$ & 3 & 9 & 15 \\
8 & 13 & 8 & 3 & 4 & 10 & 13 & 8 & 3 & $-$ & 4 & 10 \\
9 & 11 & 14 & 9 & 4 & 5 & 11 & 14 & 9 & 4 & $-$ & 5 \\
10 & 6 & 12 & 15 & 10 & 5 & 6 & 12 & 15 & 10 & 5 & $-$ \\
\end{tabular}
\normalsize
\caption{The global packing array produced from the IP algorithm for $n=5$.}
\end{table}

The wavelength arrangement of the IP algorithm can be completely characterized in the following three lemmas.

\begin{lemma}\label{lemma:GP_largelength}
For $\lceil\frac{n}{2}\rceil+1 \leq\ell\leq n$ we have
\begin{enumerate}[(a)]
\item $\d \GP(i,i+\ell) = 
	\begin{cases}
	n(n-\ell)+i+1 & \text{ if } 0\leq i\leq 2n-\ell, \\
	n(n-\ell)+i-n & \text{ if } 2n-\ell+1 \leq i \leq 2n;
	\end{cases}$
\smallskip
\item $T_{\ell}=(n+1)(n-\ell+1)$;
\smallskip
\item $\mathcal{L}_{\ell}(n(n-\ell)+i+1)$ \\ 
$\d= 
	\begin{cases}
	\{(i-n+\ell, i), (\ell+i,n+i)\} ~~\text{if } 0\leq i\leq n-\ell-1, \\
	\{(i-n+\ell, i), (\ell+i,n+i+1)\} ~~\text{if } n-\ell \leq i \leq n-1.
	\end{cases}$
\end{enumerate}
For $\lceil\frac{n}{2}\rceil+1 \leq\ell < n$ we have
\begin{enumerate}[(d)]
\item $\mathcal{L}_{\ell}(i) = \mathcal{L}_{\ell+1}(i)$, $\forall i=1,2,\ldots, n(n-\ell)$.
\end{enumerate}
\end{lemma}
\begin{proof}
These statements are proved by induction on $\ell$.
When $\ell=n$ (the first round), the paths of length $n$ are considered in the order $\langle 0,n\rangle$, $\langle n+1,0\rangle$, $\langle 1,n+1\rangle$, $\langle n+2,1\rangle$, $\ldots$, $\langle n-1,2n-1\rangle$, $\langle 2n,0\rangle$ and $\langle n,2n\rangle$.
Following the greedy strategy, it is easy to see that $\GP(0,n)=\GP(n+1,0)=1$, $\GP(1,n+1)=\GP(n+2,1)=2$, $\ldots$, $\GP(n-1,2n-1)=\GP(2n,n-1)=n$,  and $\GP(n,2n)=n+1$.
Then we have $T_n=n+1$.
Observe that for each $1\leq i\leq n$, the wavelength index $i$ is used on $2n$ edges, in which the only exceptional edge is $\{n+i-1,n+i\}$, i.e., $\mathcal{L}_{n}(i)=\{\{n+i-1,n+i\}\}$.
Then (a), (b) and (c) hold for $\ell=n$.
Consider $\ell=n-1$.
Since $\mathcal{L}_n(i)=\{\{n+i-1,n+i\}\}$ for $1\leq i\leq n$, the wavelength indices $1,2,\ldots,n$ can not be assigned to the paths of length $n-1$.
That is, the idle bands on wavelength less than or equal to $n$ will not be changed after the $(n-1)$-path round.
Hence (d) holds for $\ell=n-1$.

Assume (a), (b), (c) and (d) hold for $\ell=k$, $k\geq \lceil\frac{n}{2}\rceil+2$.
Consider $\ell=k-1$.
By induction hypotheses (c) and (d), the longest idle band on wavelength index from $1$ to $n(n-k+1)$ after the $k$-path round is of length $n-k+1$, which is smaller than $\lfloor\frac{n}{2}\rfloor-1 < k-1$.
So, the paths of length $(k-1)$ can not be assigned the wavelengths with indices less than or equal to $n(n-k+1)$. 
This implies (d): $\mathcal{L}_{k-1}(i) = \mathcal{L}_{k}(i)$, $\forall i=1,2,\ldots, n(n-k+1)$.
In the wavelength assignment, the paths of length $k-1$ are considered in the order: $\langle 0,k-1\rangle$, $\langle 2n-k+2,0\rangle$, $\langle 1,k\rangle$, $\langle 2n-k+3,1\rangle$, $\ldots$, $\langle k-2,2k-3\rangle$, $\langle 2n,k-2\rangle$, and $\langle k-1,2k-2\rangle$, $\langle k,2k-1\rangle$, $\ldots$, $\langle 2n-k+1,2n\rangle$.
By induction hypothesis (a), the wavelength index $n(n-k)+i+1$ is only occupied by the path $\langle i,i+k\rangle$ for $i=n,n+1,\ldots,2n-k$.
Then, by the greedy selection strategy, $\GP(0,k-1)=n(n-k+1)+1$ and $\GP(2n-k+2,0)=n(n-k+1)+n-k+2$, $\GP(1,k)=n(n-k+1)+2$ and $\GP(2n-k+3,1)=n(n-k+1)+n-k+3$, and so on.
As the process goes to $\GP(k-2,2k-3)=n(n-k+1)+k-1$ and $\GP(2n,k-2)=n(n-k+1)+n$, 
each wavelength index between $n(n-k+1)+1$ and $n(n-k+1)+n-k+1$ is occupied by one path of length $k-1$ and one path of length $k$, and each wavelength index between $n(n-k+1)+n-k+2$ and $n(n-k+1)+n$ is occupied by two paths of length $(k-1)$.
Consider the endpoints of these occupied paths.
Let $\mathcal{L}^*(j)$ denote the set of idle bands on wavelength index $j$ so far.
We have
\begin{equation}\label{eq:idle_band}
\begin{split}
&\mathcal{L}^*(n(n-k+1)+i+1) \\ 
&= 
	\begin{cases}
	\{(n+k+i, i), (k+i-1,n+i)\} \\ \hspace*{3cm} \text{ if } 0\leq i\leq n-k, \\
	\{(n+k+i, i), (k+i-1,n+i+1)\} \\ \hspace*{3cm} \text{ if } n-k+1 \leq i \leq n-1.
	\end{cases}
\end{split}
\end{equation}
Notice that the longest idle band in \eqref{eq:idle_band} is of length $n-k+2$, which is less than $k-1$ due to $k\geq\lfloor\frac{n}{2}\rfloor+2$.
Hence the remaining unassigned paths of length $(k-1)$ must receive wavelength index larger than $n(n-k+1)+n$.
This implies that $\mathcal{L}_{k-1}(j)=\mathcal{L}^*(j)$ for $n(n-k+1)+1\leq j\leq n(n-k+1)+n$.
Thus, condition (c) holds.
Moreover, this forces us to have $\GP(k-1,2k-2)=n(n-k+1)+k$, $\GP(k,2k-1)=n(n-k+1)+k+1$, and then up to $\GP(2n-k+1,2n)=n(n-k+1)+2n-k+2$.
Then $T_{k-1}=n(n-k+1)+2n-k+2=(n+1)(n-k+2)$, and thus (a) and (b) also hold.
Hence the result follows by induction.
\end{proof}

\medskip

\begin{lemma}\label{lemma:GP_even}
Let $n>0$ be an even integer.
We have
\begin{enumerate}[(a)]
\item $\GP(i,\frac{n}{2}+i) = \GP(\frac{n}{2}+i,n+i) = \GP(\frac{3n}{2}+1+i,i)$ $= \frac{n^2}{2}+1+i$, for $i=0,1,\ldots, \frac{n}{2}-1$;
\smallskip
\item $\GP(n+i,\frac{3n}{2}+i)=\frac{n^2}{2}-\frac{n}{2}+i$, for $i=0,1,\ldots, \frac{n}{2}$;
\smallskip
\item $\mathcal{L}_{\frac{n}{2}}(\frac{n^2}{2}+1+i) = \emptyset$, for $i=0,1,\ldots, \frac{n}{2}-1$; 
\smallskip
\item $\mathcal{L}_{\frac{n}{2}}(\frac{n^2}{2}-\frac{n}{2}+i) = \{(i,\frac{n}{2}-1+i)\}$, for $i=0,1,\ldots, \frac{n}{2}$;
\smallskip
\item $T_{\frac{n}{2}}={n+1\choose 2}$.
\end{enumerate}
\end{lemma}
\begin{proof}
By Lemma~\ref{lemma:GP_largelength}(c)--(d), the idle bands on wavelength indices less than $\frac{n^2}{2}-\frac{n}{2}$ after the $(\frac{n}{2}-1)$-path round are of length at most $\frac{n}{2}-1$, and each $\mathcal{L}_{\frac{n}{2}-1}(i)$ contains exactly one idle band of length $\frac{n}{2}$ for $\frac{n^2}{2}-\frac{n}{2}\leq i\leq \frac{n^2}{2}$.
More precisely, for $j=0,1,\ldots,\frac{n}{2}$,
\begin{equation}\label{eq:idle_even_1}
\mathcal{L}_{\frac{n}{2}+1}(\frac{n^2}{2}-\frac{n}{2}+j)=\Big\{(j,\frac{n}{2}+1+j),(n+j,\frac{3n}{2}+j)\Big\}.
\end{equation}
Consider the wavelength index $\frac{n^2}{2}+1+j$, where $0\leq j\leq\frac{n}{2}-1$.
By Lemma~\ref{lemma:GP_largelength}(a), only the edges of the path $\langle n+j,\frac{3n}{2}+j+1\rangle$ use the wavelength index $\frac{n^2}{2}+1+j$.
Hence from the greedy selection strategy we have 
\begin{equation}\label{eq:GP_even_1}
\GP(i,\frac{n}{2}+i)=\GP(\frac{3n}{2}+1+i,i)=\frac{n^2}{2}+1+i,
\end{equation}
for $i=0,1,\ldots,\frac{n}{2}-1$.
Now, the wavelength index $\frac{n^2}{2}+1+i$, $0\leq i\leq\frac{n}{2}-1$, is free on edges of the path $\langle\frac{n}{2}+i,n+i\rangle$.
This implies
\begin{equation}\label{eq:GP_even_2}
\GP(\frac{n}{2}+i,n+i)=\frac{n^2}{2}+1+i,
\end{equation}
for $i=0,1,\ldots,\frac{n}{2}-1$.
So we derive (a) by combining \eqref{eq:GP_even_1} and \eqref{eq:GP_even_2}, and then obtain (b) directly from \eqref{eq:idle_even_1}.

Since each wavelength index $\frac{n^2}{2}+1+i$, for $0\leq i\leq \frac{n}{2}-1$, is occupied by three paths of length $\frac{n}{2}$ and one path of length $\frac{n}{2}+1$, there is no idle bands on it after the $\frac{n}{2}$-path round, which implies (c).
As the idle band $(n+j,\frac{3n}{2}+j)$ in \eqref{eq:idle_even_1} is filled with a path of length $\frac{n}{2}$, we get (d).
Finally, (e) is obtained from the assignment of wavelengths to the paths of length $\frac{n}{2}$ in (a) and (b) as well as the fact that $T_{\frac{n}{2}+1}={n+1\choose 2}$ by Lemma~\ref{lemma:GP_largelength}.
Then we complete the proof.
\end{proof}	

See Figure~\ref{fig:11}(c) for an example ($n=5$) of Lemma~\ref{lemma:GP_even}.
A parallel result of Lemma~\ref{lemma:GP_even} for the case odd is $n$ is proposed below, where the proof is omitted due to the similarity between these two cases.

\begin{lemma}\label{lemma:GP_odd}
Let $n>0$ be an odd integer.
We have
\begin{enumerate}[(a)]
\item $\GP(i,i+\frac{n+1}{2})=\frac{n^2-n}{2}+i+1$, for $i=0,1,\ldots, \frac{n-3}{2}$;
\smallskip
\item $\GP(i+\frac{n-1}{2},i+n) = \GP(i+n,i+\frac{3n+1}{2}) = \GP(i+\frac{3n+1}{2},i)=\frac{n^2+1}{2}+i$, for $i=0,1,\ldots, \frac{n-1}{2}$;
\smallskip
\item $\mathcal{L}_{\frac{n+1}{2}}(\frac{n^2-n}{2}+i+1) = \{(i+\frac{n+1}{2},i+n),(i+\frac{3n+1}{2},i)\}$, for $i=0,1,\ldots, \frac{n-3}{2}$; 
\smallskip
\item $\mathcal{L}_{\frac{n+1}{2}}(\frac{n^2+1}{2}+i) = \{(i,i+\frac{n-1}{2})\}$, for $i=0,1,\ldots, \frac{n-1}{2}$;
\smallskip
\item $T_{\frac{n+1}{2}}={n+1\choose 2}$.
\end{enumerate}
\end{lemma}

\medskip

We are ready for the main result of this section. 

\begin{theorem}\label{th:cycle_odd_optimal}
For any integer $n\geq 1$, $$\Phi(C_{2n+1})= {n+1\choose 2}.$$
\end{theorem}
\begin{proof}
By Lemma~\ref{lemma:cycle_odd_lower} and Proposition~\ref{pro:upperbound}, it suffices to show that the Algorithm~\ref{alg:GGP} returns $total={n+1\choose 2}$.
After assigning wavelengths to the paths of length $\ell=n,n-1,\ldots,\lceil\frac{n}{2}\rceil$, by Lemma~\ref{lemma:GP_even}(e) and Lemma~\ref{lemma:GP_odd}(e), the maximal used wavelength index is ${n+1\choose 2}$, i.e., $T_{\lceil\frac{n}{2}\rceil}={n+1\choose 2}$.
Due to the greedy selection strategy, we only need to prove that for any pair of integers $h$ and $j$ with $1\leq h\leq\lceil\frac{n}{2}\rceil-1$ and $0\leq j\leq 2n$, there exists exactly one wavelength index $k$ such that $(j,j+h)\in\mathcal{L}_{\lceil\frac{n}{2}\rceil}(k)$ and $k\leq {n+1\choose 2}$.
As the case $n$ is odd can be dealt with in the same way, we only consider the even case.

First, the assignment of wavelengths to the paths of length $\frac{n}{2}$ shown in Lemma~\ref{lemma:GP_even}(a)--(b) states that the associated wavelength indices are less than $\frac{n^2}{2}-\frac{n}{2}$.
Then $\mathcal{L}_{\frac{n}{2}}(k)=\mathcal{L}_{\frac{n}{2}+1}(k)$ for $1\leq k<\frac{n^2}{2}-\frac{n}{2}$.
Next, by Lemma~\ref{lemma:GP_largelength}(c), the $\ell$-path round will produce $2n$ idle bands: $2n-\ell$ of them are of length $n-\ell$ and the others are of length $n-\ell+1$.
By the recursive relation in Lemma~\ref{lemma:GP_largelength}(d), in order to find all idle bands of length $h$, a fixed integer between $1$ and $\frac{n}{2}-1$, it is sufficient to consider the $(n-h)$-path and $(n-h+1)$-path rounds.
Plugging $\ell=n-h$ into Lemma~\ref{lemma:GP_largelength}(c) leads to the following:
\begin{align*}
&\mathcal{L}_{n-h}(nh+i+1) \\
&= 
	\begin{cases}
	\{(i-h, i), (n-h+i,n+i)\} \quad \text{ if } 0\leq i\leq h-1, \\
	\{(i-h, i), (n-h+i,n+i+1)\} ~\text{ if } h \leq i \leq n-1.
	\end{cases}
\end{align*}
Therefore, the $2n-(n-h)=n+h$ idle bands of length $h$ herein are $$(j,j+h),  \text{ for } -h \leq j\leq n-1.$$
Similarly, by plugging $\ell=n-h+1$ into Lemma~\ref{lemma:GP_largelength}(c), the $n-(n-h+1)=h+1$ idle bands of length $h$ produced from the $(n-h+1)$-path round are $$(j,j+h), \text{ for } n\leq j\leq 2n-h.$$
We find that for any pair of integers $h$ and $j$ with $1\leq h\leq\frac{n}{2}-1$ and $0\leq j\leq 2n$, there is exactly one idle band $(j,j+h)$ on some wavelength $k\leq {n+1\choose 2}$.
Then the result follows by applying iteratively the greedy selection strategy.
\end{proof}

One can conclude from the proof of Lemma~\ref{lemma:cycle_odd_lower} that, if there exists a path system $\mathcal{P}$ such that $\Phi(C_{2n+1},\mathcal{P})={n+1\choose 2}$, then $\mathcal{P}$ must consist of shortest paths.
Combining this fact with Theorem~\ref{th:cycle_odd_optimal} yields the following.

\begin{corollary}\label{coro:odd}
If $\mathcal{P}$ is an ideal path system of $C_{2n+1}$, then $\mathcal{P}$ is perfect.
\end{corollary}

\section{Length First Packing Algorithm}\label{sec:simulation}
The optimal global packing for odd cycles is completely characterized by means of Algorithm~\ref{alg:GGP}, IP, in Section~\ref{sec:oddcycle}.
In order to handle general graphs, we extend IP to Algorithm~\ref{alg:LFP}, \emph{Length First Packing (LFP)}.
The idea of LFP is to greedily assign the smallest free wavelength index to paths of $\mathcal{P}$ in the descending order of the path length.
In contrast to IP, there is no additional ordering relation, instead a random selection
is adopted among paths with the same length in the LFP scheme.
Clearly, IP can be viewed as a special case of LFP when the objective, $\mathcal{P}$, is the (unique) shortest path system of an odd cycle.

In order to demonstrate the power of LFP, we apply it to two classes of networks
in this section: chain networks and ring networks with random and quasi-random traffic
loads.

\begin{algorithm}
  \SetAlgoLined
  \SetKwInOut{Input}{input}\SetKwInOut{Output}{output}
  \Input{a path system, $\mathcal{P}$}
  \Output{the maximal visited wavelength index, \emph{total}}
  initialization: $total=0$, $\d\ell=\max_{P\in\mathcal{P}}||P||$ and $U=\mathcal{P}$ \;
  \While{$\ell>0$}
  {
    \For{$P\in U$}
 	{
 		\If{$||P||=\ell$}
 		{
 		$\omega(P) \gets$ the least available wavelength index\;
 		$U\gets U\setminus\{P\}$\;
 		}
 	}
 	update \emph{total} \;
    $\ell\gets \ell-1$\;
  }
    \Return \emph{total}\;
  \caption{Length First Packing (LFP)}
  \label{alg:LFP}
\end{algorithm}

\subsection{Chain networks}
A chain network (or a path) is another fundamental topology in the study of optical networks \cite{Agbinya_06}.
Let $D_n$ denote a chain of $n$ nodes.
The LFP algorithm can be used to provide an assignment which yields
the global packing number of $D_n$, $\Phi(D_n)$.

\begin{theorem}\label{thm:path}
For any integer $n\geq 1$,
$$\Phi(D_n)=\Big\lfloor\frac{n}{2}\Big\rfloor \Big\lceil\frac{n}{2}\Big\rceil.
$$
\end{theorem}
\begin{proof}
Let $\mathcal{P}$ be the unique path system of $D_n$, as there is only one path that connects any two fixed nodes.
Since the middle-most edge of $D_n$ occurs on $\lfloor\frac{n}{2}\rfloor \lceil\frac{n}{2}\rceil$ paths in $\mathcal{P}$, we have $\Phi(D_n)\geq\lfloor\frac{n}{2}\rfloor \lceil\frac{n}{2}\rceil$.

For $\ell=1,2,\ldots,n-1$, let $\mathcal{P}_{\ell}\subset\mathcal{P}$ denote the collection of paths of length $\ell$.
Obviously, $|\mathcal{P}_{\ell}|=n-\ell$.
The LFP algorithm assigns wavelengths to the path in $\mathcal{P}_{n-1}$, the two paths in $\mathcal{P}_{n-2}$, the three paths in $\mathcal{P}_{n-3}$, and so on, in a greedy fashion.
Notice that there is no addition constraint on the order of those paths which have the same length.
Denote by $T_{\ell}$ the maximal visited wavelength index after the LFP algorithm finishes the wavelength assignment of the paths in $\mathcal{P}_{\ell}$.
We claim that 
\begin{equation}\label{eq:path}
T_{\ell} = \begin{cases}
\sum_{i=\ell}^{n-1} (n-i) & \text{ if } \ell\geq\lceil\frac{n}{2}\rceil, \vspace*{7pt} \\ 
{\lceil\frac{n}{2}\rceil+1 \choose 2} + \sum_{i=\ell}^{\lceil\frac{n}{2}\rceil-1}i & \text{ if } \ell < \lceil\frac{n}{2}\rceil.
\end{cases}
\end{equation}

We only consider the case $n$ is even, as the odd case can be dealt with in the same way.
Let $n=2m$, and denote by $e$ the middle-most edge on $D_n$.
Since all paths of length $\ell\geq m$ contain $e$, they must receive distinct wavelength indices.
Then $T_\ell=|\mathcal{P}_{2m-1}|+|\mathcal{P}_{2m-2}|+\cdots+|\mathcal{P}_{\ell}|=1+2+\cdots+(2m-\ell)$.
Now, assume \eqref{eq:path} holds for $\ell+1$, where $\ell<m$, and consider the paths of length $\ell$.
In $\mathcal{P}_{\ell}$ there are exactly $\ell$ paths which contain $e$.
For each of those paths in $\mathcal{P}_{\ell}$ which do not contain $e$, it is not hard to see that there always exists at least one free wavelength index $k$ with $k\leq T_{\ell+1}$.
Therefore, $T_{\ell}=T_{\ell+1}+\ell$, and \eqref{eq:path} follows by induction.
Hence $\Phi(D_{2m})={m+1\choose 2}+{m\choose 2}=m^2$, as desired.
\end{proof}

\noindent \emph{Remark.} 
Given a chain, one can easily find a greedy wavelength assignment on its path system which does not produce its global packing number.
Take $D_6$ as an example.
Let $\{0,1,2,3,4,5\}$ denote the node set, where $a$ and $b$ are adjacent if $|a-b|=1$.
The following global packing, $\omega$, needs $10$ wavelengths, while $\Phi(D_6)=9$ due to Theorem~\ref{thm:path}.
$$
\footnotesize
\begin{array}{cccccc}
 & \omega(P_{\{1,2\}})=1 & \rightarrow & \omega(P_{\{2,3\}})=1 & \rightarrow & \omega(P_{\{4,5\}})=1 \\
\rightarrow & \omega(P_{\{0,4\}})=2 & \rightarrow & \omega(P_{\{0,5\}})=3 & \rightarrow & \omega(P_{\{0,2\}})=4 \\
\rightarrow & \omega(P_{\{3,4\}})=1 & \rightarrow & \omega(P_{\{1,5\}})=5 & \rightarrow & \omega(P_{\{0,1\}})=1 \\
\rightarrow & \omega(P_{\{3,5\}})=4 & \rightarrow & \omega(P_{\{1,4\}})=6 & \rightarrow & \omega(P_{\{0,3\}})=7 \\
\rightarrow & \omega(P_{\{1,3\}})=8 & \rightarrow & \omega(P_{\{2,4\}})=9 & \rightarrow & \omega(P_{\{2,5\}})=10 
\end{array}
\normalsize
$$
This example states that the order of paths is the essential issue, if one tries to apply a greedy algorithm to derive the global packing numbers.

\subsection{Ring networks with more random traffic load}
We originally introduced the global packing number as the minimum number of wavelengths
required to support simultaneous communication over a WDM network when the traffic
load is uniformly defined as one request per node pair.  Obviously, this traffic condition can be relaxed so that the global packing number corresponds to the minimal
number of wavelengths required for a general given traffic load for the network.
The LFP algorithm can be applied to these traffic scenarios.

First, we carry out a performance study of the LFP scheme on ring networks in comparison
with the \emph{Random Packing(RP)} scheme.  The latter scheme loops through
the wavelength set once starting with wavelength 1.  For each wavelength, $i$, there
is an iterative loop which randomly selects an unassigned path and assigns to $i$ if it does not cause any violation until no more paths can be assigned to $i$.

We considered ring networks of node sizes $n=5, 10, 15, 20, 25, 30, 35, 40$ and three kinds of traffic models defined as below:
\begin{itemize}
\item \emph{uniform}: each node pair has exactly one connection;
\item \emph{full-random}: there are $n^2$ random connections;
\item \emph{quasi-random}: based on the uniform model, there are $n$ extra random connections.
\end{itemize}
In all three models, for any node pair a shortest path connecting the nodes is chosen.
This choice is unique for odd cycles.  In the case that $n$ is even and the nodes are antipodal,
the shortest path is chosen randomly among the two candidates. 
For full-random and quasi-random traffic types, $100$ traffic models are instantiated for each
type; and for each instantiated traffic model, a simulation run consists of $10,000$ tests 
and the averaged required wavelength number is reported.

\begin{table}[h]\label{tab:uniform}
\begin{tabular}{|c||c|c|c|c|c|c|}
\hline
$n$ & $\Phi(C_n)$ & LFP & RP \\ \hline \hline 
5 & 3 & 3 & 3.47 \\ \hline
10  & 13 & 13.47 & 14.92 \\ \hline
15 & 28 & 29.69 &  33.05 \\ \hline
20 & 51 & 53.11 & 58.42 \\ \hline
25 & 78 & 82.27 & 90.28 \\ \hline
30 & 113 & 118.08 & 129.29 \\ \hline
35 & 153 & 160.31 & 174.77 \\ \hline
40 & 201 & 209.02 & 227.20 \\ \hline
\end{tabular}
\caption{Global packing performance comparison among the optimal solution, LFP and RP schemes in uniform traffic model.}
\end{table}

Table~III reports the global packing performance for the uniform traffic model.
The columns indicate node sizes, exact values of $\Phi(C_n)$ which are derived in Section~\ref{sec:evencycle} and Section~\ref{sec:oddcycle}, and the average number of required wavelengths by applying the LFP and RP schemes.
Compare to $\Phi(C_n)$, the simulation results indicate that the LFP scheme requires
$3\%$ to $6\%$ additional wavelengths, and is better than the RP scheme, which requires
an additional $13\%$ to $16\%$.

\begin{table}[h]\label{tab:random}
\begin{tabular}{|c||c|c||c|c|}
\hline
\multirow{2}{*}{$n$} &
\multicolumn{2}{c||}{full-random} &
\multicolumn{2}{c|}{quasi-random} \\
\cline{2-5}
  & LFP & RP & LFP & RP \\
\hline \hline
5 & 10.43 & 10.48 & 5.81 & 5.90 \\
\hline
10 & 34.49 & 35.06 & 17.85 & 18.61 \\
\hline
15 & 71.14 & 72.67 & 35.43 & 38.01 \\
\hline
20 & 120.94 & 124.12 & 60.55 & 64.83 \\
\hline
25 & 183.55 & 188.51 & 90.77 & 97.91 \\
\hline
30 & 258.91 & 266.59 & 128.23 & 138.18 \\
\hline
35 & 347.23 & 357.52 & 171.62 & 184.96 \\
\hline
40 & 448.05 & 462.06 & 223.26 & 240.13 \\
\hline
\end{tabular}
\caption{Global packing performance comparison between LFP and RP schemes in full- and quasi-random traffic models.}
\end{table}

Table~IV lists the average number of wavelengths required by applying LFP and RP schemes to the full-random and quasi-random traffic models.
In either case, the LFP scheme is obviously better than the RP scheme; the former uses less than $98\%$ and $93\%$ of the wavelengths that are required by the latter when $n\geq 15$ in the full-random and quasi-random scenarios, respectively.

\medskip

\section{Concluding remarks}\label{sec:conclusion}
The global packing number is a newly defined index on a graph, motivated by
the wavelength assignment problem in a WDM optical network.
This paper focuses on a ring networks, which have underlying graphs defined by cycles.
Based on graph packing method and a newly proposed greedy construction, IP algorithm, we completely characterize the global packing number of a cycle as well as its ideal path system.
Furthermore, the LFP algorithm, which can be viewed as a general version of IP algorithm, has been applied to the chain networks and networks with random traffic loads.
Simulation results show that the LFP scheme is more efficient than the random scheme in wavelength usage.  Results obtained here point to several new directions for future research,
for example, extending the concept of global packing number to networks that allow
wavelength conversion. 


\end{document}